\documentclass[submission,copyright,creativecommons]{eptcs}
\providecommand{\event}{GandALF 2016} 

\title{Distributed PROMPT-LTL Synthesis\thanks{Supported by the projects AVACS (SFB/TR~14), ASDPS (JA~2357/2--1), and TriCS (ZI~1516/1--1) of the German Research Foundation (DFG) and by the European Research Council (ERC) Grant OSARES (No.~683300).}}

\author{Swen Jacobs \qquad\qquad Leander Tentrup \qquad\qquad Martin Zimmermann
\institute{Reactive Systems Group, Saarland University, 66123 Saarbr{\"u}cken, Germany}
\email{lastname@react.uni-saarland.de}
}

\def\titlerunning{Distributed PROMPT-LTL Synthesis}
\def\authorrunning{S.\ Jacobs, L.\ Tentrup, and M.\ Zimmermann}

\usepackage{amsmath}
\usepackage{amssymb}
\usepackage{mathtools}

\usepackage{hyperref}

\usepackage{subfigure}
\usepackage{wrapfig}

\usepackage{cite}

\usepackage{tikz}
\usetikzlibrary{arrows,automata,positioning}

\usepackage[textsize=footnotesize]{todonotes}

\renewcommand{\epsilon}{\varepsilon}

\newcommand{\set}[1]{{\{#1\}}}
\newcommand{\Set}[1]{{\left\{#1\right\}}}
\newcommand{\tuple}[1]{{\langle#1\rangle}}
\renewcommand{\models}{\vDash}

\newcommand{\bigo}{\mathcal{O}}
\newcommand{\pow}[1]{2^{#1}}
\renewcommand{\phi}{\varphi}

\newcommand{\bound}{b}
\newcommand{\card}[1]{\left| {#1} \right|}

\newcommand{\nats}{\mathbb{N}}

\newcommand{\strat}[2]{(2^{#1})^* \rightarrow 2^{#2}}
\newcommand{\proj}{\mathrm{proj}}
\newcommand{\wide}{\mathrm{wide}}
\newcommand{\fun}[2]{#1 \rightarrow #2}

\newcommand{\tsys}{\mathcal{S}} 

\newcommand{\nbw}{\mathcal{N}}
\newcommand{\ucw}{\mathcal{U}}
\newcommand{\uct}{\mathcal{U}_T}

\newcommand{\bwin}{B}
\newcommand{\Bwin}{\mathcal{B}}
\newcommand{\cobwin}{\reflectbox{$B$}}

\newcommand{\true}{\mathbf{tt}}
\newcommand{\false}{\mathbf{ff}}
\renewcommand{\implies}{\mathbin{\rightarrow}}
\DeclareMathOperator\F{\mathbf{F}}
\DeclareMathOperator\G{\mathbf{G}}
\DeclareMathOperator\GF{\mathbf{GF}}
\DeclareMathOperator\FG{\mathbf{FG}}
\DeclareMathOperator\GFp{\mathbf{GF_P}}
\DeclareMathOperator\Fp{\mathbf{F_P}}
\newcommand{\U}{\mathbin{\mathbf{U}}}
\DeclareMathOperator\X{\mathbf{X}}
\newcommand{\R}{\mathbin{\mathbf{R}}}

\newcommand{\rel}{\mathit{rel}}
\newcommand{\alt}{\mathit{alt}}

\newcommand{\ap}{\mathrm{AP}} 

\newcommand{\sched}{\mathit{sched}} 
\newcommand{\Sched}{\mathrm{Sched}}

\newcommand{\ltl}{\mathrm{LTL}}

\newcommand{\pltl}{\mathrm{PLTL}}
\newcommand{\prompt}{\mathrm{PROMPT}$\textendash$\ltl}

\newcommand{\arch}{\mathcal{A}}
\newcommand{\penv}{p_\mathit{env}}
\newcommand{\pminus}{P^-}
\newcommand{\distprod}{\otimes}
\newcommand{\Distprod}{\bigotimes}

\newcommand{\bgraph}{G}   
\newcommand{\cbgraph}{G}  

\newcommand{\nlogspace}{\textsc{NLogSpace}}
\newcommand{\pspace}{\textsc{PSpace}}
\newcommand{\twoexp}{\textsc{2ExpTime}}

\usepackage{amsthm}

\newtheorem{lemma}{Lemma}

\newtheorem{theorem}{Theorem}
\newtheorem{corollary}{Corollary}

\theoremstyle{definition}
\newtheorem{example}{Example}
\newtheorem{remark}{Remark}

\begin{document}
\maketitle

\begin{abstract}
  We consider the synthesis of distributed implementations for specifications in Prompt Linear Temporal Logic ($\prompt$), which extends $\ltl$ by temporal operators equipped with parameters that bound their scope.
  For single process synthesis it is well-established that such parametric extensions do not increase worst-case complexities.
  
  For synchronous systems, we show that, despite being more powerful, the distributed realizability problem for $\prompt$ is not harder than its $\ltl$ counterpart.
  For asynchronous systems we have to consider an assume-guarantee synthesis problem, as we have to express scheduling assumptions. As asynchronous distributed synthesis is already undecidable for $\ltl$, we give a semi-decision procedure for the $\prompt$ assume-guarantee synthesis problem based on bounded synthesis.
\end{abstract}

\section{Introduction}

Linear Temporal Logic~\cite{Pnueli77} ($\ltl$) is the most prominent specification language for reactive systems and the basis for industrial languages like ForSpec~\cite{Forspec02} and PSL~\cite{EisnerFismanPSL}. Its advantages include a compact variable-free syntax and intuitive semantics as well as the exponential compilation property, which explains its attractive algorithmic properties: every $\ltl$ formula can be translated into an equivalent B\"uchi automaton of exponential size. This yields a polynomial space model checking algorithm and a doubly-exponential time algorithm for solving two-player games. Such games solve the monolithic $\ltl$ synthesis problem: given a specification, construct a correct-by-design implementation.

However, $\ltl$ lacks the ability to express timing constraints. For example, the request-response property~$\G(\mathit{req} \rightarrow \F \mathit{resp})$ requires that every request~$\mathit{req}$ is eventually responded to by a $\mathit{resp}$. It is satisfied even if the waiting times between requests and responses diverge, i.e., it is impossible to require that requests are granted within a fixed, but arbitrary, amount of time. While it is possible to encode an a-priori fixed bound for an eventually into $\ltl$, this requires prior knowledge of the system's granularity and incurs a blow-up when translated to automata, and is thus considered impractical.

To overcome this shortcoming of $\ltl$, Alur et al.\ introduced parametric $\ltl$~($\pltl$)~\cite{journals/tocl/AlurETP01}, which extends $\ltl$ with parameterized operators of the form $\F_{\le x}$ and $\G_{\le y}$, where $x$ and $y$ are variables. The formula~$\G(req \rightarrow \F_{\le x}\, resp)$ expresses  that every request is answered within an arbitrary, but fixed, number of steps~$\alpha(x)$. Here, $\alpha$ is a variable valuation, a mapping of variables to natural numbers. Typically, one is interested in whether a $\pltl$ formula is satisfied with respect to some variable valuation, e.g., model checking a transition system~$\tsys$ against a $\pltl$ specification~$\varphi$ amounts to determining whether there is an $\alpha$ such that every trace of $\tsys$ satisfies $\varphi$ with respect to $\alpha$. Alur et al.\ showed that the $\pltl$ model checking problem is $\pspace$-complete. Due to monotonicity of the parameterized operators, one can assume that all variables $y$ in parameterized always operators $\G_{\le y}$ are mapped to zero, as variable valuations are quantified existentially in the problem statements. Dually, again due to monotonicity, one can assume that all variables $x$ in parameterized eventually operators $\F_{\le x}$ are mapped to the same value, namely the maximum of the bounds. Thus, in many cases the parameterized always operators and different variables for parameterized eventually operators are not necessary. 

Motivated by this, Kupferman et al.\ introduced $\prompt$~\cite{journals/fmsd/KupfermanPV09}, which can be seen as the fragment of $\pltl$ without the parameterized always operator and with a single bound~$k$ for the parameterized eventually operators. They proved that $\prompt$ model checking is $\pspace$-complete and solving $\prompt$ games is $\twoexp$-complete, i.e., not harder than $\ltl$ games. While the results of Alur et al.\ rely on involved pumping arguments, the results of Kupferman et al.\ are all based on the so-called alternating color technique, which basically allows to reduce $\prompt$ to $\ltl$. 
Furthermore, the result on $\prompt$ games was extended to $\pltl$ games~\cite{journals/tcs/Zimmermann13}, again using the alternating color technique. 
These results show that adding parameters to $\ltl$ does not increase the asymptotic complexity of the model checking and the game-solving problem, which is still true for even more expressive logics~\cite{journals/corr/FaymonvilleZ14,DBLP:journals/corr/Zimmermann15a}.

The synthesis problems mentioned above assume a setting of complete information, i.e., every part of the system has a complete view on the system as a whole. However, this setting is highly unrealistic in virtually any system.
Distributed synthesis on the other hand, is the problem of synthesizing multiple components with incomplete information.
Since there are specifications that are not implementable, one differentiates synthesis from the corresponding decision problem, i.e., the \emph{realizability} problem of a formal specification.
We focus on the latter, but note that from the methods presented here, implementations are  efficiently extractable from a proof of realizability.

The realizability problem for distributed systems dates back to work of Pnueli and Rosner in the early nineties~\cite{conf/focs/PnueliR90}.
They showed that the realizability problem for LTL becomes undecidable already for the simple architecture of two processes with pairwise different inputs.
In subsequent work, it was shown that certain classes of architectures, like pipelines and rings, can still be synthesized automatically~\cite{conf/lics/KupfermanV01, DBLP:conf/fsttcs/MohalikW03}.
Later, a complete characterization of the architectures for which the realizability problem is decidable was given by Finkbeiner and Schewe by the \emph{information fork} criterion~\cite{conf/lics/FinkbeinerS05}.
Intuitively, an architecture contains an information fork, if there is an information flow from the environment to two different processes where the information to one process is hidden from the other and vice versa.
The distributed realizability problem is decidable for all architectures without information fork.
Beyond decidability results, semi-algorithms like bounded synthesis~\cite{journals/sttt/FinkbeinerS13} give an architecture-independent synthesis method that is particularly well-suited for finding small-sized implementations.

\textbf{Our Contributions.} As mentioned above, one can add parameters to $\ltl$ for free: the complexity of the model checking problem and of solving infinite games does not increase. This raises the question whether this observation also holds for the distributed realizability of parametric temporal logics.
For synchronous systems, we can answer this question affirmatively. 
For every class of architectures with decidable $\ltl$ realizability, the $\prompt$ realizability problem is decidable, too.
To show this, we apply the alternating color technique~\cite{journals/fmsd/KupfermanPV09} to reduce the distributed realizability problem of $\prompt$ to the one of $\ltl$: one can again add parameterized operators to $\ltl$ for~free.

For asynchronous systems, the environment is typically assumed to take over the responsibility for the scheduling decision~\cite{conf/lopstr/ScheweF06}.
Consequently, the resulting schedules may be unrealistic, e.g., one process may not be scheduled at all.
While \emph{fairness} assumptions such as ``every process is scheduled infinitely often'' solve this problem for $\ltl$ specifications, they are insufficient for $\prompt$: a fair scheduler can still delay process activations arbitrarily long and thereby prevent the system from satisfying its $\prompt$ specification for any bound~$k$. \emph{Bounded fair} scheduling, where every process is guaranteed to be scheduled in bounded intervals, overcomes this problem.
Since bounded fairness can be expressed in $\prompt$, the realizability problem in asynchronous architectures can be formulated more generally as an assume-guarantee realizability problem that consists of two $\prompt$ specifications. 
We give a semi-decision procedure for this problem based on a new method for checking emptiness of two-colored B\"uchi graphs~\cite{journals/fmsd/KupfermanPV09} and an extension of bounded synthesis~\cite{journals/sttt/FinkbeinerS13}.
As asynchronous $\ltl$ realizability for architectures with more than one process is undecidable~\cite{conf/lopstr/ScheweF06}, the same result holds for $\prompt$ realizability.
Decidability in the one process case, which holds for $\ltl$~\cite{conf/lopstr/ScheweF06}, is left open. 

All these results also hold for $\pltl$ and even stronger logics~\cite{journals/corr/FaymonvilleZ14,DBLP:journals/corr/Zimmermann15a} to which the alternating color technique is still applicable. 

\textbf{Related Work.}
There is a rich literature regarding the synthesis of distributed systems from global $\omega$-regular specifications~\cite{conf/focs/PnueliR90,conf/lics/KupfermanV01,conf/lics/FinkbeinerS05,conf/fmcad/ChatterjeeHOP13,journals/ipl/Schewe14, DBLP:conf/fsttcs/MohalikW03}.
We are not aware of work that is concerned with the realizability of parameterized logics in this setting.
For local specifications, i.e., specifications that only relate the inputs and outputs of single processes, the realizability problem becomes decidable for a larger class of architectures~\cite{conf/icalp/MadhusudanT01}.
An extension of these results to context-free languages was given by Fridman and Puchala~\cite{journals/acta/FridmanP14}.
The realizability problem for asynchronous systems and LTL specifications is undecidable for architectures with more than one process to be synthesized~\cite{conf/lopstr/ScheweF06}.
Later, Gastin et al.\ showed decidability of a restricted specification language and certain types of architectures, i.e., well-connected~\cite{journals/fmsd/GastinSZ09} and acyclic~\cite{journals/tocl/GastinS13} ones.
Bounded synthesis~\cite{journals/sttt/FinkbeinerS13} provides a flexible synthesis framework that can be used for synthesizing implementations for both the asynchronous and synchronous setting.

\section{Prompt LTL}\label{prompt}

Throughout this work, we fix a set~$\ap$ of atomic propositions. The formulas of $\prompt$
are given by the grammar
\begin{equation*}\phi \Coloneqq a \mid \neg a \mid \phi \wedge \phi \mid \phi \vee
\phi
  \mid \X \phi \mid \phi \U \phi \mid \phi \R \phi \mid
  \Fp \phi
  \enspace,\end{equation*}
where $a \in \ap$ is an atomic proposition, $\neg,\wedge,\vee$ are the usual boolean operators, and $\X, \U, \R$ are $\ltl$ operators next, until, and release. 
We use the derived operators
$\true \coloneqq a \vee \neg a$ and $\false \coloneqq a \wedge \neg a$ for some fixed $a \in \ap$,
and $\F \phi \coloneqq \true \U \phi$ and $\G \phi \coloneqq \false \R \phi$ as usual.
Furthermore, we use $\phi \implies \psi$ as shorthand for $\neg \phi \vee \psi$,
if the antecedent~$\phi$ is a (negated) atomic
proposition (where we identify $\neg \neg a$ with $a$).
We define the size of $\phi$ to be the number of subfomulas of $\phi$.
The satisfaction relation for $\prompt$ is defined between an $\omega$-word~$w = w_0 w_1 w_2 \cdots \in \left( \pow{\ap} \right)^{ \omega }$, a
position~$n$ of $w$, a bound~$k$ for the prompt-eventually operators, and a~$\prompt$ formula. For the $\ltl$ operators, it is defined as usual (and oblivious to $k$) and for the prompt-eventually we have
\begin{itemize}
	\item $(w,n,k)\models\Fp\phi$ if, and only if, there exists a $j$ in
the range
$0\le j \le k$ such that $(w,n+j,k)\models\phi$.

\end{itemize}
For the sake of brevity, we write $(w,k) \models \phi$ instead of
$(w,0,k) \models \phi$ and say that $w$ is a model of $\phi$ with
respect to $k$. Note that $(w, n, k) \models \phi$ implies $(w, n, k') \models \phi$ for every $k' \ge k$, i.e., satisfaction with respect to $k$ is an upwards-closed property.

\paragraph[The Alternating Color Technique]{\bf The Alternating Color Technique.} \label{subsection_altcolor}

In this subsection, we recall the alternating color technique, which 
Kupferman et al.\ introduced to solve model checking, assume-guarantee model checking, and the
realizability problem for $\prompt$ specifications~\cite{journals/fmsd/KupfermanPV09}.

Let $r\notin \ap$ be a fixed fresh proposition. An
$\omega$-word~$w'\in\left(2^{\ap\cup\{r\}}\right)^{\omega}$ is an $r$-coloring of
$w\in\left(2^{\ap}\right)^{\omega}$ if $w_n'\cap \ap=w_n$, i.e., $w_n$ and $w_n'$
coincide on all propositions in $\ap$. The additional proposition~$r$ can be
thought of as the color of $w_n'$: we say
that the \emph{ color changes} at position~$n$, if $n=0$ or if the truth values of $r$ in $w_{n-1}'$
and in $w_n'$ are not equal. In this situation, we say that $n$ is a \emph{change point}.
An \emph{$r$-block} is a maximal infix~$w_m' \cdots w_{n}'$ of $w'$ such that the color changes at $m$ and $n+1$, but not in between. 

Let $k \ge 1$: we say
that $w'$ is \emph{$k$-spaced} if the color changes infinitely often and each
$r$-block has length at least $k$; we say that $w'$ is \emph{$k$-bounded}, if each
$r$-block has length at most $k$. Note that $k$-boundedness implies that the color changes
infinitely often.

Given a $\prompt$ formula~$\phi$, let $\rel_r(\phi)$ denote the formula obtained by inductively replacing every subformula~$\Fp\psi$ by
\begin{equation*}
(r\implies (r\U(\neg r\U \rel_r(\psi))))\wedge(\neg r\implies (\neg r\U(
r\U \rel_r(\psi))))\enspace,
\end{equation*}
which is only linearly larger than $\phi$ and requires every prompt eventually to be satisfied within at most one color change (not counting the position where $\psi$ holds). 
Furthermore, the formula
$\alt_r = \GF r\wedge \GF\neg r$ is satisfied if the colors change infinitely
often. Finally, we define the $\ltl$ formula~$c_r(\phi) = \rel_r(\phi) \wedge \alt_r$.
Kupferman et al.\ showed that $\phi$ and $c_r(\phi)$ are in some sense equivalent on $\omega$-words
which are bounded and spaced. 

\begin{lemma}[Lemma~2.1 of \cite{journals/fmsd/KupfermanPV09}]
\label{lemma_alternatingcolor}
Let $\phi$ be a $\prompt$ formula, and let $w \in \left( \pow{\ap} \right)^{ \omega }$.
\begin{enumerate}
  \item \label{lemma_alternatingcolor_pltltoltl}
If $(w,k)\models \phi$, then $w' \models c_r(\phi)$ for every $k$-spaced $r$-coloring~$w'$ of~$w$.
  
  \item \label{lemma_alternatingcolor_ltltopltl}
If $w'$ is a $k$-bounded $r$-coloring of $w$ with $w' \models c_r(\phi)$, then $(w,2k)\models\phi$.
\end{enumerate}
\end{lemma}
Whenever possible, we drop the subscript~$r$ for the sake of readability, if $r$ is clear from context. However, when we consider asynchronous systems in Section~\ref{sec:asynchronous_distributed_synthesis}, we need to relativize two formulas with different colors, which necessitates the introduction of the subscripts.

\section{Synchronous Distributed Synthesis} \label{sec:synchronous_distributed_synthesis}

$\prompt$ specifications can give guarantees that $\ltl$ cannot, for example by asserting not only that requests to a system are answered \emph{eventually}, but also that there is an \emph{upper bound} on the reaction time.
This is especially important in distributed systems, since such timing constraints become more difficult to implement because of information flows between the various parts of the system.

Consider for example a distributed computation system, where
a central master gets \emph{important} and \emph{unimportant} tasks, and can forward tasks to a number of clients. 
A client can either enqueue the task, which means that it will be processed \emph{eventually}, or clear the client-side queue and process the task immediately. The latter operation is very costly (we have to remember the open tasks as they still need to be completed), but guarantees an upper bound on the completion time.
While in $\ltl$ we can only specify that all incoming tasks are processed eventually, in $\prompt$ we can specify that the answer time to important tasks is bounded by the formula
$\G (\mathit{important\text{-}task} \rightarrow \Fp \mathit{finished\text{-}task})$.\footnote{A similar constraint could be simulated in $\ltl$ by writing that on every important incoming task, the worker queues are cleared. This, however, removes implementation freedom and requires the developer to determine how to implement the feature, instead of letting the synthesis algorithm decide.}

We continue by formalizing the distributed realizability problem.
Let $X$ and $Y$ be finite and pairwise disjoint sets of variables.
A \emph{valuation} of $X$ is a subset of $X$; thus, the set of all valuations of $X$ is $\pow{X}$.
For $w = w_0 w_1 w_2 \cdots \in (\pow{X})^\omega$ and $w'=w'_0 w'_1 w'_2 \cdots \in (\pow{Y})^\omega$, let $w \cup w' = (w_0 \cup w'_0) (w_1 \cup w'_1) (w_2 \cup w'_2) \cdots \in (\pow{X \cup Y})^\omega$. 

\paragraph{Strategies.}

A \emph{strategy} $f \colon \strat{X}{Y}$ maps a history of valuations of $X$ to a valuation of $Y$.
A \emph{$\pow{Y}$-labeled $\pow{X}$-transition system} $\tsys$ is a tuple $\tuple{S,s_0,\Delta,l}$ where
$S$ is a finite set of states,
$s_0 \in S$ is the designated initial state,
$\Delta \colon \fun{S \times \pow{X}}{S}$ is the transition function, and
$l \colon \fun{S}{\pow{Y}}$ is the state-labeling.
We generalize the transition function to sequences over $\pow{X}$ by defining $\Delta^* \colon \fun{(\pow{X})^*}{S}$ recursively as $\Delta^*(\epsilon) = s_0$ and $\Delta^*(w_0 \cdots w_{n-1} w_n) = \Delta(\Delta^*(w_0 \cdots w_{n-1}), w_n)$ for $w_0 \cdots w_{n-1} w_n \in (\pow{X})^+$.
A transition system $\tsys$ \emph{generates} the strategy $f$ if $f(w) = l(\Delta^*(w))$ for every $w \in (\pow{X})^*$.
A strategy $f$ is called \emph{finite-state} if there exists a transition system that generates $f$.

Let $X'$ and $Y'$ be finite and disjoint sets where $X'$ is additionally disjoint from $Y$ and $Y'$ is additionally pairwise disjoint from $X$ and $Y$.
Further, let $f\colon \strat{X}{Y}$ and $f'\colon \strat{X}{Y'}$ be two strategies with the same domain but pairwise different co-domain $\pow{Y}$ and $\pow{Y'}$.
The \emph{product} $f \times f' \colon \strat{X}{Y \cup Y'}$ of $f$ and $f'$ is defined as $(f \times f')(w) = f(w) \cup f'(w)$ for every $w \in (\pow{X})^*$.
The $\pow{X}$-projection of a sequence $w_0 \cdots w_n \in (2^{X \cup X'})^*$ is $\proj_{\pow{X}}(w_0 \cdots w_n) = (w_0 \cap X) \cdots (w_n \cap X) \in (\pow{X})^*$. 
The $\pow{X'}$-widening of a strategy $f \colon \strat{X}{Y}$ is defined as $\wide_{\pow{X'}}(f) \colon \strat{X \cup X'}{Y}$ with $\wide_{\pow{X'}}(f)(w) = f(\proj_{\pow{X}}(w))$ for $w \in (2^{X \cup X'})^*$.
For strategies $f \colon \strat{X}{Y}$ and $f'\colon \strat{X'}{Y'}$, the \emph{distributed product} $f \distprod f' \colon \strat{X \cup X'}{Y \cup Y'}$ is defined as the product $\wide_{2^{X' \setminus X}}(f) \times \wide_{2^{X \setminus X'}}(f')$.

The behavior of a strategy $f\colon \strat{X}{Y}$ is characterized by an infinite tree that branches by the valuations of $X$ and whose nodes $w \in (\pow{X})^*$ are labeled with the strategic choice $f(w)$.
For an infinite word $w = w_0 w_1 w_2 \cdots \in (\pow{X})^\omega$, the corresponding labeled path is defined as $(f(\epsilon) \cup w_0)(f(w_0) \cup w_1)(f(w_0 w_1) \cup w_2)\cdots \in (2^{X \cup Y})^\omega$.
We lift the set containment operator $\in$ to the containment of a labeled path $w = w_0 w_1 w_2 \cdots \in (2^{X \cup Y})^\omega$ in a strategy tree induced by $f \colon \strat{X}{Y}$, i.e., $w \in f$ if, and only if, $f(\epsilon) = w_0 \cap Y$ and $f((w_0 \cap X) \cdots (w_i \cap X)) = w_{i+1} \cap Y$ for all $i \geq 0$.
We define the satisfaction of a $\prompt$ formula $\varphi$ (over propositions $X \cup Y$) on strategy~$f$ with respect to the bound~$k$, written $(f,k) \models \phi$ for short, as $(w,k) \models \varphi$ for all paths $w \in f$.

\paragraph{Distributed Systems.}

We characterize a distributed system as a set of processes with a fixed communication topology, called an \emph{architecture} in the following.
Recall that $\ap$ is the set of atomic propositions used to build formulas.
An \emph{architecture} $\arch$ is a tuple $\tuple{P,\penv,\{I_p\}_{p \in P}, \{O_p\}_{p \in P}}$, where $P$ is the finite set of processes and $p_\mathit{env} \in P$ is the distinct environment process. We denote by $\pminus = P \setminus \set{\penv}$ the set of system 
processes.

Given a process $p \in P$, the inputs and outputs of this process are $I_p \subseteq \ap$ and $O_p \subseteq \ap$, respectively, where we assume~$I_{\penv} = \emptyset$.
We use the notation $I_{P'}$ and $O_{P'}$ for some $P' \subseteq P$ for $\bigcup_{p \in P'} I_p$ and $\bigcup_{p \in P'} O_p$, respectively. While processes may share the same inputs (in case of broadcasting), the outputs of processes must be pairwise disjoint, i.e., for all $p \neq p' \in P$ it holds that $O_p \cap O_{p'} = \emptyset$.

An \emph{implementation} of a process $p \in \pminus$ is a strategy $f_p \colon \strat{I_p}{O_p}$ mapping finite input sequences to a valuation of the output variables.

\begin{example}

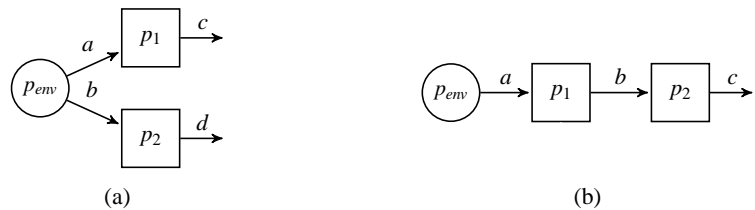
\begin{figure}[h]
\centering
\subfigure[]{
  \begin{tikzpicture}[->,>=stealth',shorten >=1pt,auto,node distance=1cm,semithick,scale=0.8,transform shape]
  
  \tikzstyle{every state}=[shape=rectangle]
  \tikzstyle{envstate}=[shape=circle,scale=0.95]
  
  \node[state,envstate]           (env) {$p_\mathit{env}$};
  \node[state,above right=0 and 1 of env] (P0)  {$p_1$};
  \node[state,below right=0 and 1 of env] (P1)  {$p_2$};
  
  \path (P0.east) edge node  {$c$} +(right:0.75)
        (P1.east) edge node  {$d$} +(right:0.75)
        (env) edge node [pos=0.65]{$a$} (P0)
        (env) edge node [pos=0.2]{$b$} (P1)
        ;


\end{tikzpicture}
  \label{fig:independent_architecture}
}\qquad\qquad\qquad%
\subfigure[]{
  \begin{tikzpicture}[->,>=stealth',shorten >=1pt,auto,node distance=1cm,semithick,scale=0.8,transform shape]
  
  \tikzstyle{every state}=[shape=rectangle]
  \tikzstyle{envstate}=[shape=circle,scale=0.95]
  
  \node[state,envstate] (env)      {$p_\mathit{env}$};
  \node[state] (P0) [right=0.85 of env] {$p_1$};
  \node[state] (P1) [right=of P0]  {$p_2$};
  
  \path (env) edge node {$a$} (P0)
        (P0) edge node {$b$} (P1)
        (P1.east) edge node {$c$} +(right:0.75)
        ;
  
  \path[white] (P0.south) edge node {} +(down:0.75);

\end{tikzpicture}
  \label{fig:pipeline_architecture}
}
\caption[]{Example architectures}
\label{fig:architectures}
\end{figure}

Figure~\ref{fig:architectures} shows example architectures $\arch_1$ and $\arch_2$, where
\begin{align*}
  \arch_1 ={} & \tuple{
    \set{\penv, p_1, p_2},
    \penv,
    \set{ \penv \to \emptyset, p_1 \to \set{a}, p_2 \to \set{b} },
    \set{ \penv \to \set{a, b}, p_1 \to \set{c}, p_2 \to \set{d} }
  }, \text{ and} \\
  \arch_2 ={} & \tuple{
    \set{\penv, p_1, p_2},
    \penv,
    \set{ \penv \to \emptyset, p_1 \to \set{a}, p_2 \to \set{b} },
    \set{ \penv \to \set{a}, p_1 \to \set{b}, p_2 \to \set{c} }
  } \enspace.
\end{align*}
The architecture $\arch_1$ in Fig.~\ref{fig:independent_architecture} contains two system processes, $p_1$ and $p_2$, and the environment process $\penv$.
The processes $p_1$ and $p_2$ receive the inputs~$a$, respectively $b$, from the environment and output $c$ and $d$, respectively.
Hence, the environment can provide process $p_1$ with information that is hidden from $p_2$ and vice versa.
In contrast, architecture $\arch_2$, depicted in Fig.~\ref{fig:pipeline_architecture}, is a pipeline architecture where information from the environment can only propagate through the pipeline processes $p_1$ and $p_2$.

\end{example}

\paragraph{Distributed Realizability.}

Let $\arch = \tuple{P,\penv,\{I_p\}_{p \in P}, \{O_p\}_{p \in P}}$ be an architecture.
The \emph{distributed realizability problem for $\arch$} is to decide, given a $\prompt$ formula $\varphi$, whether there exist a bound $k$ and a finite-state implementation $f_p$ for every process $p \in P^-$, such that the distributed product $\Distprod_{p \in P^-} f_p$ satisfies $\varphi$ with respect to $k$, i.e., $(\Distprod_{p \in P^-} f_p,k) \models \varphi$. In this case, we say that $\phi$ is realizable in $\arch$.
The distributed realizability problem for $\ltl$ is a special case, as $\ltl$ is a fragment of $\prompt$.

Let $r \notin \ap$ be the fresh proposition introduced for the alternating color technique to relativize formulas and let $\arch = \tuple{P,\penv,\{I_p\}_{p \in P}, \{O_p\}_{p \in P}}$ be an architecture as above. We define the architecture $\arch^r$ as $\tuple{P \cup \set{p_r},\penv,\{I_p\}_{p \in P} \cup \set{I_r},\{O_p\}_{p \in P} \cup \set{O_r}}$, where $I_r = \emptyset$ and $O_r = \set{r}$.
Intuitively, this describes an architecture where one additional process $p_r$ is responsible for providing sequences in $(\pow{\set{r}})^\omega$, i.e., a coloring by $r$. We show that $\varphi$ in $\arch$ and $c_r(\varphi)$  in $\arch^r$ are equi-realizable by applying the alternating color technique. As the processes are synchronized, the proof is similar to the one for the single-process case by Kupferman et al.~\cite{journals/fmsd/KupfermanPV09}.

\begin{theorem}
  A $\prompt$ formula $\varphi$ is realizable in $\arch$ if, and only if, $c_r(\varphi)$ is realizable in $\arch^r$.
\end{theorem}
\begin{proof}%
  Let $\arch = \tuple{P,\penv,\{I_p\}_{p \in P}, \{O_p\}_{p \in P}}$ be an architecture and $\varphi$ be a $\prompt$ formula.
  
  Assume that the $\prompt$ formula $\varphi$ is realizable in $\arch$.
  Then, there exist finite-state strategies $f_p$ for $p \in P^-$ and a bound $k$ satisfying the $\prompt$ distributed realizability problem $\tuple{\arch,\varphi}$.
  For every $w \in \Distprod_{p \in P^-} f_p$, it holds that $(w,k) \models \varphi$.
  By Lemma~\ref{lemma_alternatingcolor}.\ref{lemma_alternatingcolor_pltltoltl} it holds that every $k$-spaced $r$-coloring $w'$ of $w$ satisfies $c_r(\varphi)$.
  Let $f_r \colon (2^\emptyset)^* \rightarrow 2^\set{r}$ be a (finite-state) strategy that produces the $k$-spaced sequence $(\emptyset^k \set{r}^k)^\omega$.
  Then, the process implementations $\set{f_p}_{p \in P^-}$ together with $f_r$ are a solution to the $\ltl$ distributed realizability problem $\tuple{\arch^r,c_r(\varphi)}$.

  Now, assume that the $\ltl$ formula $c_r(\varphi)$ is realizable in the architecture~$\arch^r$.
  Thus, there exist finite-state strategies $f_p$ for $p \in P^-$ and a finite-state strategy $f_r$ for process~$p_r$. Note that the strategy~$f_r \colon \strat{\emptyset}{\set{r}}$ has a unique output~$w_r \in (\pow{{\set{r}}})^\omega$, as it has no inputs. We claim that $w_r$ is $k$-bounded, where $k$ is the number of states of the transition system~$\tsys = \tuple{S,s_0,\Delta,l}$ generating $f_r$. To see this, note that $f_r$ has no inputs, i.e., every state of $\tsys$ has a unique successor in $\Delta$, and the unique run of $\tsys$ on $\emptyset^\omega$ ends up in a loop which is traversed ad infinitum. As the output~$w_r$ has infinitely many change points, the loop contains at least one state~$s$ labeled by $l(s) = \emptyset$ and at last one state~$s'$ with $l(s') = \set{r}$. Thus, the maximal length of a block of $w_r$ is bounded by the length of the loop, which in turn is bounded by the size of $\tsys$. 

  Hence, for every $w \in \Distprod_{p \in P^-} f_p$, the word $w_r \cup w$ is a $k$-bounded $r$-coloring of $w$ with $w_r \cup w \models \rel_r(\varphi)$.
  By Lemma~\ref{lemma_alternatingcolor}.\ref{lemma_alternatingcolor_ltltopltl}, for all such $w$ it holds that $(w,2k) \models \varphi$.
  Hence, $\set{f_p}_{p \in P^-}$ together with the bound~$2k$ is a solution to the $\prompt$ distributed realizability problem.
\end{proof}

To conclude, we show that the newly introduced process $p_r$ preserves the \emph{information fork} criterion~\cite{conf/lics/FinkbeinerS05}.
Formally, consider tuples $\tuple{P', V', p, p'}$, where $P'$ is a subset of the processes, $V'$ is a subset of the variables disjoint from $I_p \cup I_{p'}$, and $p,p' \in  P^- \setminus P'$ are two different processes.
Such a tuple is an information fork in $\arch$ if $P'$ together with the edges that are labeled with at least one variable from $V'$ forms a sub-graph of $\arch$ rooted in the environment and there exist two nodes $q, q' \in P'$ that have edges to $p, p'$, respectively, such that $O_\set{q,p} \nsubseteq I_{p'}$ and $O_\set{q',p'} \nsubseteq I_p$.
For example, the architecture in Fig.~\ref{fig:independent_architecture} contains the information fork $(\set{\penv}, \emptyset, p_1, p_2)$, while the pipeline architecture depicted in Fig.~\ref{fig:pipeline_architecture} does not contain an information fork.

\begin{lemma}
  $\arch^r$ contains an information fork if, and only if, $\arch$ contains an information fork.
\end{lemma}
\begin{proof}
  The \emph{if} direction follows immediately by construction: if $\tuple{P', V', p, p'}$ is an information fork in $\arch$ then it is an information fork in $\arch^r$ as well.
  Hence, assume $\tuple{P', V', p, p'}$ is an information fork in $\arch^r$.
  It holds that neither $p_r = p$ nor $p_r = p'$ since $p_r$ has no incoming edges.
  As $I_{p_r} = \emptyset$, $p_r$ cannot be in a sub-graph that is rooted in the environment, hence, $p_r \notin P'$ and $r \notin V'$.
  It follows that $\tuple{P', V', p, p'}$ is an information fork in $\arch$.
\end{proof}

Thus, we can use well-known results for the decidability of distributed realizability for $\ltl$ and weakly ordered architectures~\cite{conf/lics/FinkbeinerS05}, i.e., those without an information fork.

\begin{corollary}
  Let $\arch$ be an architecture.
  The $\prompt$ distributed realizability problem for $\arch$ is decidable if, and only if, $\arch$ is weakly ordered.
\end{corollary}\noindent
Furthermore, we can directly apply semi-algorithms for the distributed realizability problem, such as bounded synthesis~\cite{journals/sttt/FinkbeinerS13}, to effectively construct small-sized solutions.

\section{Asynchronous Distributed Synthesis} \label{sec:asynchronous_distributed_synthesis}

The asynchronous system model is a generalization of the synchronous model discussed in the last section.
In an asynchronous system, not all processes are scheduled at the same time.
We model the scheduler as part of the environment, i.e., at any given time the environment additionally signals whether a process is enabled.
The resulting distributed realizability problem is already undecidable for $\ltl$ specifications and systems with more than one process~\cite{conf/lopstr/ScheweF06}.

We have to adapt the definition of the $\prompt$ realizability problem for the asynchronous setting.
Using the definition from Section~\ref{sec:synchronous_distributed_synthesis}, the system can never satisfy a $\prompt$ formula if the scheduler is part of the environment, since it may delay scheduling indefinitely. Moreover, even if the scheduler is assumed to be fair, it can still build increasing delay blocks between process activation times, such that it is impossible for the system to guarantee any bound $k \in \nats$.
Hence, we employ the concept of \emph{bounded fair} schedulers and allow the system valuations to depend on the scheduler bound.
More generally, this is a typical instance of an assume-guarantee specification: under the assumption that the scheduler is bounded fair, the system satisfies its specification.
In the following, we formally introduce the distributed realizability problem for asynchronous systems and assume-guarantee specifications.

To model scheduling, we introduce an additional set $\Sched = \set{\sched_p \mid p \in \pminus}$ of atomic propositions. The valuation of $\sched_p$ indicates whether system process~$p$ is currently scheduled or not.
Given a (synchronous) architecture $\arch = \tuple{P,\penv,\{I_p\}_{p \in P}, \{O_p\}_{p \in P}}$, we define the asynchronous architecture $\arch^*$ as the architecture with the environment output $O^*_{\penv} = O_{\penv} \cup \Sched$.
Furthermore, we extend the input $I_p$ of a process by its scheduling variable $\sched_p$, i.e., $I_p^* = I_p \cup \set{\sched_p}$ for every $p \in P^-$.
The environment can decide in every step which processes to schedule.
When a process is not scheduled, its \emph{state}---and thereby its outputs---do not change~\cite{journals/sttt/FinkbeinerS13}.
Formally, let $f_p$ for $p \in P^-$ be a finite-state implementation for a process~$p$ and $\tsys_p = \tuple{S,s_0,\Delta,l}$ a transition system that generates~$f_p$.
For every path $w = w_0 w_1 w_2 \cdots \in (\pow{I_p^*})^\omega$ it holds that if $\sched_p \notin w_i$ for some $i \in \nats$, then $\Delta^*(w[i]) = \Delta^*(w[i+1])$, where $w[i]$ denotes the prefix~$w_0 w_1 \cdots w_{i}$ of $w$.

A $\prompt$ assume-guarantee specification~$\tuple{\varphi,\psi}$ consists of a pair of $\prompt$ formulas. The asynchronous assume-guarantee realizability problem asks, given an asynchronous architecture~$\arch^*$ and $\tuple{\varphi,\psi}$ as above, whether there exists a finite-state implementation~$f_p$ for every process $p \in P^-$ such that for every bound~$k$ there is a bound $l$ such that for every $w \in \Distprod_{p \in P^-} f_p$, we have that $(w, k) \models \varphi$ implies $(w, l) \models \psi$.
In this case, we say that $\Distprod_{p \in P^-} f_p$ satisfies $\tuple{\varphi,\psi}$.

Consider the bounded fairness specification discussed above, which is expressed by the formula~$\varphi = \bigwedge_{p \in \pminus} \GFp \sched_p$, i.e., for every point in time, every $p$ is scheduled within a bounded number of steps.
That is, we use $\varphi$ as an assumption on the environment which implies that the guarantee $\psi$ only has to be satisfied if $\varphi$ holds.
Consider for example the asynchronous architecture corresponding to Fig.~\ref{fig:independent_architecture} and the $\prompt$ specification $\psi = \G (\Fp c \land \Fp \neg c \land \Fp d \land \Fp \neg d)$.
Even when we assume a fair scheduler, i.e., $\varphi = \GF \sched_{p_1} \land \GF \sched_{p_2}$, the environment can prevent one process from satisfying the specification for any bound $l$.
This problem is fixed by assuming the scheduler to be bounded fair, i.e., $\varphi = \GFp \sched_{p_1} \land \GFp \sched_{p_2}$.
Then, there exist realizing implementations for processes $p_1$ and $p_2$ (that alternate between enabling and disabling the output), and the bound on the guarantee is $l = 2 \cdot k$ for every bound~$k$.

Unlike $\ltl$, where the assume-guarantee problem $\tuple{\varphi,\psi}$ can be reduced to the $\ltl$ realizability problem for the implication $\varphi \rightarrow \psi$, this is not possible in $\prompt$ due to the  quantifier alternation on the bounds.
Indeed, it is still open whether the $\prompt$ assume-guarantee realizability problem in the single-process case is decidable.
We show that even if the problem turns out to be decidable, an implementation that realizes the specification may need in general infinite memory.

\begin{lemma}
   There exists an assume-guarantee $\prompt$ specification that can be realized with an infinite-state strategy, but not with a finite-state strategy.
\end{lemma}
\begin{proof}
  Consider the assume-guarantee specification $\tuple{\varphi,\psi}$ with $\varphi = \GFp o \lor \FG \neg o$ and $\psi = \false$ and a single process architecture with $I = \emptyset$ and $O = \set{o}$.
  As the guarantee $\psi$ is false, the implementation has to falsify the assumption $\varphi$ for every bound $k$ on the prompt-eventually operator to realize $\tuple{\varphi,\psi}$.
  To falsify $\varphi$ with respect to $k$, the implementation has to produce a sequence $w \in (2^\set{o})^\omega$ where $o$ is repeatedly true and where $\emptyset^k$ is an infix  of $w$.
  Thus, the size of the implementation depends on $k$ and an implementation that falsifies $\varphi$ for every $k$ must have infinite memory.
\end{proof}
Since the $\ltl$ realizability problem is undecidable and implementations for $\prompt$ assume-guarantee specifications may need infinite memory, the $\prompt$ assume-guarantee realizability problem for asynchronous architectures may be at best solvable by a semi-decision procedure.
We present such a semi-algorithm for the asynchronous distributed realizability problem for assume-guarantee $\prompt$ specifications based on bounded synthesis~\cite{journals/sttt/FinkbeinerS13}.
In bounded synthesis, a transition system of a fixed size is ``guessed'' and model checked by a constraint solver.
Model checking for $\prompt$ can be solved by checking pumpable non-emptiness of colored B\"uchi graphs~\cite{journals/fmsd/KupfermanPV09}.
However, the pumpability condition cannot directly be expressed in the bounded synthesis constraint system.
Hence, in Section~\ref{sec:colored-buchi-graphs}, we give an alternative solution to the non-emptiness of colored B\"uchi graphs by a reduction to B\"uchi graphs that have access to the state space of the transition system.
We use this result to build the semi-algorithm that is presented in Section~\ref{sec:semi-algorith-ag}.

\subsection{Nonemptiness of Colored B\"uchi Graphs} \label{sec:colored-buchi-graphs}

In the case of $\ltl$ specifications, the nonemptiness problem for B\"uchi graphs gives a classical solution to the model checking problem for a given system $\tsys$.
Let $\varphi$ be the $\ltl$ formula that $\tsys$ should satisfy.
In a preprocessing step, the negation of $\varphi$ is translated to a nondeterministic B\"uchi word automaton $\nbw_{\neg \varphi}$~\cite{BaierKatoen08}.
Then $\varphi$ is violated by $\tsys$ if, and only if, the B\"uchi graph $G$ representing the product of $\tsys$ and $\nbw_{\neg\varphi}$ is nonempty.
An accepting path $\pi$ in $G$ witnesses a computation of $\mathcal{S}$ that violates~$\varphi$.
\emph{Colored B\"uchi graphs} are an extension to those graphs in the context of model checking $\prompt$~\cite{journals/fmsd/KupfermanPV09}.

A colored B\"uchi graph of degree two is a tuple $G = \tuple{\set{r,r'},V,E,v_0,L,\Bwin}$ where
$r$ and $r'$ are propositions,
$V$ is a set of vertices,
$E \subseteq V \times V$ is a set of edges,
$v_0 \in V$ is the designated initial vertex,
$L\colon\fun{V}{\pow{\set{r,r'}}}$ describes the color of a vertex, and
$\Bwin = \set{\bwin_1,\bwin_2}$ is a generalized B\"uchi condition of index two, i.e., $\bwin_1, \bwin_2 \subseteq V$.
A B\"uchi graph is a special case where we omit the labeling function and are interested in finding an accepting path.
A path $\pi = v_0 v_1 v_2 \cdots \in V^\omega$ is pumpable, if we can pump all its $r'$-blocks without pumping its $r$-blocks.
Formally, a path is pumpable if for all adjacent $r'$-change points $i$ and $i'$, there are positions $j$, $j'$, and $j''$ such that $i \leq j < j ' < j'' < i'$, $v_j = v_{j''}$ and $r \in L(v_j)$ if, and only if,  $r \notin L(v_{j'})$.
A path $\pi$ is accepting, if it visits both $\bwin_1$ and $\bwin_2$ infinitely often.
The \emph{pumpable nonemptiness} problem for $G$ is to decide whether $G$ has a pumpable accepting path.
It is $\nlogspace$-complete and solvable in linear time~\cite{journals/fmsd/KupfermanPV09}.

We give an alternative solution to this problem based on a reduction to the nonemptiness problem of B\"uchi graphs.
To this end, we construct a non-deterministic safety automaton $\nbw_\text{pump}$ that characterizes the pumpability condition.
Note that an infinite word is accepted by a safety automaton if, and only if, there exists an infinite run on this word.
\begin{lemma} \label{thm:buchi_pumpable}
  Let $\cbgraph = \tuple{\set{r,r'},V,E,v_0,L,\Bwin}$ be a colored B\"uchi graph of degree two.
  There exists a B\"uchi graph $\bgraph'$ with $\bigo(\card{\bgraph'}) = \bigo(\card{\cbgraph}^2)$ such that $\cbgraph$ has a pumpable accepting path if, and only if, $\bgraph'$ has an accepting path.
\end{lemma}
\begin{proof}
  We define a non-deterministic safety automaton $\nbw_\text{pump} = \tuple{V \times 2^\set{r,r'},S,s_0,\delta,S}$ over the alphabet $V \times \pow{\set{r,r'}}$ that checks the pumpability condition.
  The product of $\cbgraph$ and $\nbw_\text{pump}$ (defined later) represents the B\"uchi graph $G'$ where every accepting path is pumpable.
  
  The language $\mathcal{L} \subseteq (V \times \pow{\set{r,r'}})^\omega$ of pumpable paths (with respect to a fixed set of vertices $V$) is an $\omega$-regular language that can be recognized by a small non-deterministic safety automaton.
  This automaton~$\nbw_\text{pump}$ operates in 3 phases between every pair of adjacent $r'$-change points:
  first, it non-deterministically remembers a vertex $v$ and the corresponding truth value of $r$.
  Then, it checks that this value changes and thereafter it remains to show that the vertex $v$ repeats before the next $r'$-change point.
  Thus, the state space~$S$ of $\nbw_\text{pump}$ is
  \begin{align*}
    \set{s_0} \cup \Set{s_{v,x} \mid v \in V, x \in 2^\set{r,r'}}
    \cup \Set{s'_{v,y} \mid v \in V, y \in 2^\set{r,r'}}
    \cup \Set{s''_z \mid z \in 2^\set{r'}}
  \end{align*}
  and the initial state is $s_0$.
  The state space corresponds to the 3 phases: In the states $s_{v,x}$ a vertex $v$ and a truth value of $r$ are remembered, before state $s'_{v,y}$ the value of $r$ changes, and $s''_z$ is the state after the vertex repetition.
  The transition function $\delta \colon (S \times (V \times 2^\set{r,r'})) \rightarrow 2^S$ is defined as follows:
  \begin{itemize}
    \item $\delta(s_0,(v,x)) =  \set{ s_{v,x} }$
    \item $\delta(s_{v,x},(v',x')) \ni \begin{cases}
      s_{v,x} & \text{if } x =_\set{r'} x' \\
      s_{v',x'} & \text{if } x =_\set{r'} x' \\
      s'_{v,x'} & \text{if } x =_\set{r'} x' \text{ and } x \neq_\set{r} x'
    \end{cases}$
    \item $\delta(s'_{v,y},(v',x)) \ni \begin{cases}
      s'_{v,y} & \text{if } x =_\set{r'} y \text{ and } v' \neq v \\
      s''_{y \cap \set{r'}} & \text{if } x =_\set{r'} y \text{ and } v = v
    \end{cases}$
    \item $\delta(s''_z, (v,x)) \ni \begin{cases}
      s''_z & \text{if } x =_\set{r'} z \\
      s_{v,x} & \text{if } x \neq_\set{r'} y
    \end{cases}$
  \end{itemize}
  where $A =_C B$ is defined as $(A \cap C) = (B \cap C)$.
  The size of $\nbw_\text{pump}$ is in $O(\card{V})$.
  Figure~\ref{fig:safety_automaton_pumpable} gives a visualization of this automaton.
  
  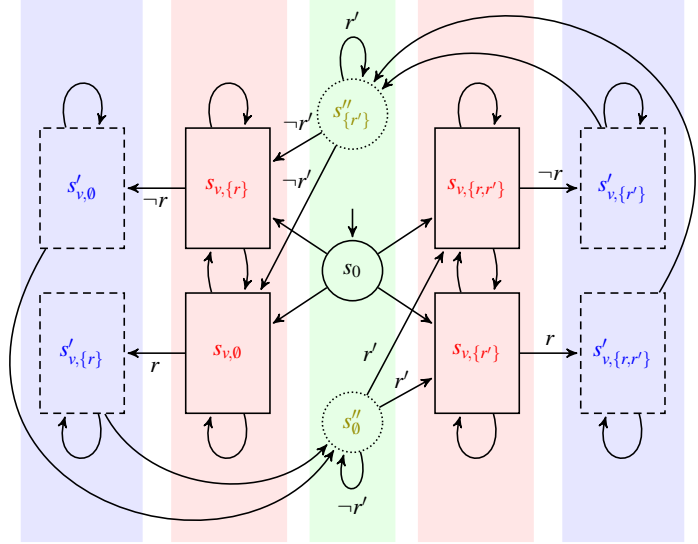
\begin{figure}[h]
    \centering
    \begin{tikzpicture}[->,>=stealth',shorten >=1pt,auto,node distance=1cm,semithick,scale=0.8,transform shape]

  \draw[fill,color=red!10] (-1.1,4.5) rectangle (-3,-4.5);
  \draw[fill,color=red!10] (1.1,4.5) rectangle (3,-4.5);
  
  \draw[fill,color=blue!10] (-3.5,4.5) rectangle (-5.5,-4.5);
  \draw[fill,color=blue!10] (3.5,4.5) rectangle (5.5,-4.5);
  
  \draw[fill,color=green!10] (-0.7,4.5) rectangle (0.7,-4.5);

  \tikzstyle{block}=[state,rectangle,minimum width=1.4cm,minimum height=2cm]
  \tikzstyle{nonblock}=[state,minimum size=1cm]

  \node[nonblock,initial above,initial text=] (init) {$s_0$};
  
  \node[block,above right=0 and 1 of init] (s_qp) {\color{red}$s_{v,\set{r,r'}}$};
  \node[block,below right=0 and 1 of init] (s_q) {\color{red}$s_{v,\set{r'}}$};
  \node[block,above left=0 and 1 of init] (s_p) {\color{red}$s_{v,\set{r}}$};
  \node[block,below left=0 and 1 of init] (s_) {\color{red}$s_{v,\emptyset}$};
  
  \node[block,right=of s_qp,densely dashed] (sp_q) {\color{blue}$s'_{v,\set{r'}}$};
  \node[block,right=of s_q,densely dashed] (sp_qp) {\color{blue}$s'_{v,\set{r,r'}}$};
  \node[block,left=of s_p,densely dashed] (sp_) {\color{blue}$s'_{v,\emptyset}$};
  \node[block,left=of s_,densely dashed] (sp_p) {\color{blue}$s'_{v,\set{r}}$};
  
  \node[nonblock,above=1.5 of init,densely dotted] (spp_q) {\color{olive}$s''_\set{r'}$};
  \node[nonblock,below=1.5 of init,densely dotted] (spp_) {\color{olive}$s''_\emptyset$};
  
  \draw (init) edge (s_qp)
        (init) edge (s_q)
        (init) edge (s_p)
        (init) edge (s_)
        
        (s_qp) edge[bend left=15] (s_q)
        (s_q) edge[bend left=15] (s_qp)
        (s_p) edge[bend left=15] (s_)
        (s_) edge[bend left=15] (s_p)
        
        (s_qp) edge node {$\neg r$} (sp_q)
        (s_q) edge node {$r$} (sp_qp)
        (s_p) edge node {$\neg r$} (sp_)
        (s_) edge node {$r$} (sp_p)
        
        (sp_q) edge[bend right=55] (spp_q)
        (sp_p) edge[bend right=55] (spp_)
        
        (spp_q) edge[loop above] node {$r'$} ()
        (spp_) edge[loop below] node {$\neg r'$} ()
        
        (spp_q) edge node[pos=0,swap,yshift=-3pt] {$\neg r'$} (s_p)
        (spp_q) edge node[pos=0.2,swap,yshift=-10pt] {$\neg r'$} (s_)
        (spp_) edge node[near start,yshift=-5pt] {$r'$} (s_qp)
        (spp_) edge node[pos=0.7,yshift=-5pt] {$r'$} (s_q)
        
        (s_qp) edge[loop above,min distance=10mm] (s_qp)
        (s_p) edge[loop above,min distance=10mm] (s_p)
        (s_q) edge[loop below,min distance=10mm] (s_q)
        (s_) edge[loop below,min distance=10mm] (s_)
        
        (sp_qp) edge[loop below,min distance=10mm] (sp_qp)
        (sp_p) edge[loop below,min distance=10mm] (sp_p)
        (sp_q) edge[loop above,min distance=10mm] (sp_q)
        (sp_) edge[loop above,min distance=10mm] (sp_)        
        ;
  
  \draw (sp_qp) .. controls +(3,5) and +(2,3) .. (spp_q);
  \draw (sp_) .. controls +(-3,-5) and +(-2,-3) .. (spp_);
  
\end{tikzpicture}
    \caption{Schematic visualization of the automaton $\nbw_\text{pump}$ from the proof of Lemma~\ref{thm:buchi_pumpable}. The 3 phases are clearly visible: In the red states {\color{red}$s_{v,x}$} (solid rectangles) the values $(v,x)$ are non-deterministically stored and those states can only be left if there is a change in the value of $r$. The subsequent blue states {\color{blue}$s'_{v,y}$} (dashed rectangles) can only be left in case of a vertex repetition leading to the green state {\color{olive}$s''_z$} (dotted circles) that waits for the next $r'$ change point.}
    \label{fig:safety_automaton_pumpable}
  \end{figure}
  \begin{remark}
    Note that in the context of this proof, it would be enough to remember a vertex $v$ without the valuation of $\set{r,r'}$ as the vertex determines the valuation by the labeling function $L\colon\fun{v}{2^\set{r,r'}}$ of $\cbgraph$.
    However, we will later use $\nbw_\text{pump}$ in a more general setting (cf.~Section~\ref{sec:semi-algorith-ag}).
  \end{remark}
  
  We define the product $G'$ of the colored B\"uchi graph $\cbgraph = \tuple{\set{r,r'},V,E,v_0,L,\Bwin}$ and the automaton~$\nbw_\text{pump}$ as the B\"uchi graph $(V \times S,E',(v_0,s_0),\Bwin')$, where 
  \begin{equation*}
    ((v,s),(v',s')) \in E' \quad\Leftrightarrow\quad (v,v') \in E \land s' \in \delta(s,(v, L(v)))  
  \end{equation*}
  and where $\Bwin' = (\bwin'_1,\bwin'_2)$ is given by  $\bwin'_i = \set{(v,s) \mid v \in \bwin_i \text{ and } s \in S}$ for  $i \in \set{1,2}$.
  The size of $G'$ is in $\bigo(\card{G}^2)$. It remains to show that $G$ has a pumpable accepting path if, and only if, $G'$ has an accepting path.

  Consider a pumpable accepting path $\pi$ in $G$.
  We show that there is a corresponding accepting path $\pi'$ in $G'$.
  Let $i$ and $i'$ be adjacent $r'$-change points.
  Then there are positions $j$, $j'$, and $j''$ such that $i \leq j < j ' < j'' < i'$, $v_j = v_{j''}$ and $r \in L(v_j)$ if, and only if,  $r \notin L(v_{j'})$.
  By construction, at position $i$, automaton $\nbw_\text{pump}$ is some state from the set $\set{s_0,s''_\emptyset,s''_\set{r'}}$.
  We follow the automaton and remember vertex $v$ and the truth value of $r$ at position $j \geq i$ (some state $s_{v,x}$).
  Next, we take the transition to $s'_{v,y}$ where the truth value of $r$ changes (at position $j'$).
  Lastly, we check that there is a vertex repetition (at position $j''$) and go to state $s''_z$.
  At the next $r'$-change point $i'$, the argument repeats.
  This path is accepting, as the original one is accepting.
  
  Now, consider an accepting path $\pi$ in $G'$.
  We show that there is a pumpable accepting path in $G$.
  Let $\pi'$ be the projection of every position of $\pi$ to the first component.
  By construction, $\pi'$ is an accepting path in $G$.
  Let $\pi_i \pi_{i+1} \cdots \pi_{i'}$ be an $r'$-block of $\pi$.
  As $\pi$ has a run on automaton $\nbw_\text{pump}$, we know that there exists a state repetition between $i$ and $i'$ where the truth value of $r$ changes in between.
  Hence, the path $\pi'$ is pumpable. 
\end{proof}

\subsection{A Semi-Algorithm for Assume-Guarantee Realizability} \label{sec:semi-algorith-ag}
As the assume-guarantee realizability problem for asynchronous architectures is undecidable and infinite-state strategies are required in general, we give a semi-decision procedure for the problem, as an extension of the bounded synthesis approach~\cite{journals/sttt/FinkbeinerS13}. 
Based on an $\ltl$ specification $\varphi$, an architecture $\arch$, and a size bound $\bound$, bounded synthesis separately considers the problems of finding a global transition system that satisfies the given specification, and of dividing the transition system into local components according to the given architecture. To this end, two sets of constraints are generated: an encoding of the satisfaction of $\varphi$ by a global transition system $\tsys$ of size $\bound$, and an encoding of the architectural constraints that divides this global system into local components. If the conjunction of both sets of constraints is satisfiable, then a model of the constraints represents a distributed system that satisfies $\varphi$ in $\arch$. Since the architectural constraints we consider are the same as in standard bounded synthesis, we only have to modify the constraints encoding the existence of a global transition system that satisfies the given specification.

In the following, we use the techniques developed in the last subsection to generalize the encoding of the specification from a single $\ltl$ formula $\varphi$ to an assume-guarantee specification $\tuple{\varphi,\psi}$ in $\prompt$. 
Given an assume-guarantee specification $\tuple{\varphi,\psi}$,
we first solve the problem of model-checking assume-guarantee specifications by building a universal co-B\"uchi tree automaton $\uct$ that accepts a transition system $\tsys$ if, and only if, $\tsys$ satisfies $\tuple{\varphi,\psi}$.
From $\uct$ and a given bound $\bound$, we then build a constraint system that is satisfiable if, and only if, an implementation $\tsys$ of $\tuple{\varphi,\psi}$ with size $\bound$ exists. Finally, the encoding of architectural constraints can be adopted without changes from the original approach to obtain a conjunction of constraints that is satisfiable if, and only if, there is a system of size $\bound$ that satisfies $\tuple{\varphi,\psi}$ in $\arch$.

\paragraph{Encoding $\tuple{\varphi,\psi}$ into B\"uchi automata.}
Let $\arch^* = \tuple{P,\penv,\{I^*_p\}_{p \in P}, \{O^*_p\}_{p \in P}}$ be  an asynchronous architecture and let $I = O^*_{\penv}$ and $O = \bigcup_{p \in P^-} O^*_p $ be the set of inputs, respectively outputs, of the composition of the system processes.
First, we construct the non-deterministic B\"uchi  automaton $\nbw_{\overline{c}_{r'}(\psi) \land c_r(\varphi)} = \tuple{2^{I \cup O  \cup \set{r,r'}}, Q, q_0, \delta, \bwin}$, where $\overline{c}_{r'}(\psi) = \alt_{r'} \land \neg\rel_{r'}(\psi)$ whose language contains exactly those paths that satisfy ${\overline{c}_{r'}(\psi) \land c_r(\varphi)}$~\cite{BaierKatoen08}.
\begin{lemma}[cf.~Theorem~6.2 of \cite{journals/fmsd/KupfermanPV09}] \label{thm:model checking-asynchronous}
  Let $\tsys$ be a $2^O$-labeled $2^I$-transition system.
  Then $\tsys$ does not satisfy $\tuple{\psi,\varphi}$ if, and only if, the product of $\tsys$ and $\nbw_{\overline{c}_{r'}(\psi) \land c_r(\varphi)}$ is pumpable non-empty.
\end{lemma}
To check the existence of pumpable error paths, we use the non-deterministic automaton~$\nbw_\text{pump} = \tuple{V \times 2^\set{r,r'},S,s_0,\delta',S}$ from the proof of Lemma~\ref{thm:buchi_pumpable}. Here, we let $V = X \times Q$, where $X$ is a set with $\bound$ elements, representing the state space of the desired solution $\tsys$, and $Q$ is the state space of the automaton~$\nbw_{\overline{c}_{r'}(\psi) \land c_r(\varphi)}$ defined above. 
That is, we use as $V$ the state space $X \times Q$ of the colored B\"uchi graph that is used to model check an implementation $\tsys$ against a specification $\tuple{\psi,\varphi}$. 

The product of $\nbw_{\overline{c}_{r'}(\psi) \land c_r(\varphi)}$ and $\nbw_\text{pump}$ is an automaton $\nbw$ that operates on the inputs $I$, outputs $O$, propositions $\set{r,r'}$, and the state space $X$ of the implementation, and accepts all those paths that are pumpable and violate the assume-guarantee specification (cf. Lemma~\ref{thm:buchi_pumpable}). 

$\nbw$ is defined as 
$$\tuple{2^{I  \cup O \cup \set{r,r'}} \times X, Q \times S, (q_0,s_0),\delta^*,\bwin^*},$$
where $\delta^* \colon Q \times S \times 2^{I \cup O \cup \set{r,r'}} \times \set{x} \rightarrow 2^{Q \times S}$ is defined as 
$$\delta^*((q,s),(\sigma,x)) = \left\{ (q',s') ~\mid~ q' \in \delta(q, \sigma) ~\wedge~ s' \in \delta'(s, \set{q,x} \cup (\sigma\cap\set{r,r'})) \right\},$$
and $\bwin^*$ is the B\"uchi condition $\set{(q,s) \mid q \in \bwin, s \in S}$.

We complement $\nbw$, resulting in a universal co-B\"uchi automaton $\ucw$ that accepts a given sequence $w \in (2^{I \cup  \set{r,r'}})^\omega$ of inputs and the behavior of an implementation $\tsys$ on $w$ iff the execution of $\tsys$ on $w$ satisfies $\tuple{\psi,\varphi}$. Finally, we construct a universal co-B\"uchi tree automaton $\uct = (2^O \times X, 2^{I \cup \set{r,r'}},Q,q_0,\delta,\cobwin)$ by spanning a copy of $\ucw$ for every direction in $2^{I \cup \set{r,r'}}$.
Then, an implementation $\tsys$ is accepted by $\uct$ if, and only if, $\tsys$ satisfies $\tuple{\varphi,\psi}$ (for all possible input sequences). Thus,  
$\uct$ solves the problem of model-checking assume-guarantee specifications.

\paragraph{Encoding the automaton into constraints.}
Now, we use a slightly modified bounded synthesis algorithm~\cite{journals/sttt/FinkbeinerS13} to encode $\uct$ into a set of constraints in a first-order theory with uninterpreted functions and a total order, such that the constraints are satisfiable iff there exists an implementation $\tsys$ that satisfies $\tuple{\varphi,\psi}$. The main difference to the existing approach is that the specification automaton $\uct$ has access to the states of the implementation $\tsys$. This is not a problem, since the generated constraints explicitly refer to the state space of $\tsys$ anyway. The original proof of correctness can be used with minor modifications to obtain the following corollary.

\begin{corollary}
Given an assume-guarantee specification $\tuple{\varphi,\psi}$ and a bound $b$, there is a constraint system (in a decidable first-order theory) that is satisfiable if, and only if, there exist an implementation $\tsys$ of size~$b$ such that $\tsys$ satisfies $\tuple{\varphi,\psi}$.
\end{corollary}

\paragraph{Encoding of architectural constraints.}
As mentioned above, the encoding of architectural constraints can be adopted without changes, and it can in particular also contain additional bounds on the state space of every single component the conjunction of both sets of constraints then asks for the existence of a distributed implementation $\tsys = \Distprod_{p \in P^-} f_p$ of size $\bound$ that satisfies $\tuple{\varphi,\psi}$, possibly with additional bounds $b_{p}$ for every $p \in \pminus$ on the size of the components. Thus, we obtain:

\begin{theorem}
  Given an assume-guarantee specification $\tuple{\varphi,\psi}$, an asynchronous architecture $\arch^*$, and a family of bounds $b_{p}$ for every $p \in \pminus$, there is a constraint system (in a decidable first-order theory) that is satisfiable if, and only if, there exist implementations $f_{p}$ of size~$b_p$ for every $p \in \pminus$ such that $\Distprod_{p \in P^-} f_p$ satisfies $\tuple{\varphi,\psi}$ in $\arch^*$.
\end{theorem}

By exhaustively traversing the space of bounds~$(b_{p})_{p \in \pminus}$ and by solving the resulting constraint system as in the previous theorem, we obtain a semi-algorithm for the asynchronous assume-guarantee realizability problem for $\prompt$. Furthermore, this also solves the synthesis problem, as implementations are efficiently obtained from a satisfying assignment of the constraint system. 

\begin{corollary}
  Let $\arch$ be an asynchronous architecture.
  The $\prompt$ distributed assume-guarantee realizability problem for $\arch$ is semi-decidable.
\end{corollary}

\section{Conclusion}

In this paper, we have initiated the investigation of distributed synthesis for parameterized specifications, in particular for $\prompt$. This logic subsumes $\ltl$, but additionally allows to express bounded satisfaction of system properties, instead of only eventual satisfaction. To the best of our knowledge, this is the first treatment of $\prompt$ specifications in distributed synthesis.

We have shown that for the case of synchronous distributed systems, we can reduce the $\prompt$ synthesis problem to an $\ltl$ synthesis problem. Thus, the complexity of $\prompt$ synthesis corresponds to the complexity of $\ltl$ synthesis, and the $\prompt$ realizability problem is decidable if, and only if,  the $\ltl$ realizability problem is decidable. For the case of asynchronous distributed systems with multiple components, the $\prompt$ realizability problem is undecidable, again corresponding to the result for $\ltl$. For this case, we give a semi-decision procedure based on a novel method for checking emptiness of two-colored B\"uchi~graphs. All these results also hold for $\pltl$ and the even stronger logics from~\cite{journals/corr/FaymonvilleZ14,DBLP:journals/corr/Zimmermann15a}, as they have the exponential compilation property and as the alternating coloring technique is applicable to these logics as well.

Among the problems that remain open is realizability of $\prompt$ specifications in asynchronous distributed systems with a single component. This problem can be reduced to the (single-process) assume-guarantee realizability problem for $\prompt$, which was left open in~\cite{journals/fmsd/KupfermanPV09}.

\bibliographystyle{eptcs}
\bibliography{main}

\begin{thebibliography}{10}
\providecommand{\bibitemdeclare}[2]{}
\providecommand{\surnamestart}{}
\providecommand{\surnameend}{}
\providecommand{\urlprefix}{Available at }
\providecommand{\url}[1]{\texttt{#1}}
\providecommand{\href}[2]{\texttt{#2}}
\providecommand{\urlalt}[2]{\href{#1}{#2}}
\providecommand{\doi}[1]{doi:\urlalt{http://dx.doi.org/#1}{#1}}
\providecommand{\bibinfo}[2]{#2}

\bibitemdeclare{article}{journals/tocl/AlurETP01}
\bibitem{journals/tocl/AlurETP01}
\bibinfo{author}{Rajeev \surnamestart Alur\surnameend}, \bibinfo{author}{Kousha
  \surnamestart Etessami\surnameend}, \bibinfo{author}{Salvatore \surnamestart
  {La Torre}\surnameend} \& \bibinfo{author}{Doron \surnamestart
  Peled\surnameend} (\bibinfo{year}{2001}): \emph{\bibinfo{title}{Parametric
  temporal logic for "model measuring"}}.
\newblock {\sl \bibinfo{journal}{{ACM} Trans. Comput. Log.}}
  \bibinfo{volume}{2}(\bibinfo{number}{3}), pp. \bibinfo{pages}{388--407},
  \doi{10.1145/377978.377990}.

\bibitemdeclare{inproceedings}{Forspec02}
\bibitem{Forspec02}
\bibinfo{author}{Roy \surnamestart Armoni\surnameend}, \bibinfo{author}{Limor
  \surnamestart Fix\surnameend}, \bibinfo{author}{Alon \surnamestart
  Flaisher\surnameend}, \bibinfo{author}{Rob \surnamestart Gerth\surnameend},
  \bibinfo{author}{Boris \surnamestart Ginsburg\surnameend},
  \bibinfo{author}{Tomer \surnamestart Kanza\surnameend},
  \bibinfo{author}{Avner \surnamestart Landver\surnameend},
  \bibinfo{author}{Sela \surnamestart Mador-Haim\surnameend},
  \bibinfo{author}{Eli \surnamestart Singerman\surnameend},
  \bibinfo{author}{Andreas \surnamestart Tiemeyer\surnameend},
  \bibinfo{author}{Moshe~Y. \surnamestart Vardi\surnameend} \&
  \bibinfo{author}{Yael \surnamestart Zbar\surnameend} (\bibinfo{year}{2002}):
  \emph{\bibinfo{title}{The {ForSpec} Temporal Logic: A New Temporal
  Property-Specification Language}}.
\newblock In \bibinfo{editor}{Joost-Pieter \surnamestart Katoen\surnameend} \&
  \bibinfo{editor}{Perdita \surnamestart Stevens\surnameend}, editors: {\sl
  \bibinfo{booktitle}{TACAS 2002}}, {\sl \bibinfo{series}{LNCS}}
  \bibinfo{volume}{2280}, \bibinfo{publisher}{Springer}, pp.
  \bibinfo{pages}{296--311}, \doi{10.1007/3-540-46002-0\_21}.

\bibitemdeclare{book}{BaierKatoen08}
\bibitem{BaierKatoen08}
\bibinfo{author}{Christel \surnamestart Baier\surnameend} \&
  \bibinfo{author}{Joost-Pieter \surnamestart Katoen\surnameend}
  (\bibinfo{year}{2008}): \emph{\bibinfo{title}{Principles of Model Checking}}.
\newblock \bibinfo{publisher}{The MIT Press}.

\bibitemdeclare{inproceedings}{conf/fmcad/ChatterjeeHOP13}
\bibitem{conf/fmcad/ChatterjeeHOP13}
\bibinfo{author}{Krishnendu \surnamestart Chatterjee\surnameend},
  \bibinfo{author}{Thomas~A. \surnamestart Henzinger\surnameend},
  \bibinfo{author}{Jan \surnamestart Otop\surnameend} \&
  \bibinfo{author}{Andreas \surnamestart Pavlogiannis\surnameend}
  (\bibinfo{year}{2013}): \emph{\bibinfo{title}{Distributed synthesis for {LTL}
  fragments}}.
\newblock In: {\sl \bibinfo{booktitle}{{FMCAD} 2013}},
  \bibinfo{publisher}{IEEE}, pp. \bibinfo{pages}{18--25},
  \doi{10.1109/FMCAD.2013.6679386}.

\bibitemdeclare{book}{EisnerFismanPSL}
\bibitem{EisnerFismanPSL}
\bibinfo{author}{Cindy \surnamestart Eisner\surnameend} \&
  \bibinfo{author}{Dana \surnamestart Fisman\surnameend}
  (\bibinfo{year}{2006}): \emph{\bibinfo{title}{A Practical Introduction to
  PSL}}.
\newblock \bibinfo{series}{Integrated Circuits and Systems},
  \bibinfo{publisher}{Springer}, \doi{10.1007/978-0-387-36123-9}.

\bibitemdeclare{inproceedings}{journals/corr/FaymonvilleZ14}
\bibitem{journals/corr/FaymonvilleZ14}
\bibinfo{author}{Peter \surnamestart Faymonville\surnameend} \&
  \bibinfo{author}{Martin \surnamestart Zimmermann\surnameend}
  (\bibinfo{year}{2014}): \emph{\bibinfo{title}{Parametric Linear Dynamic
  Logic}}.
\newblock In: {\sl \bibinfo{booktitle}{{GandALF} 2014}}, {\sl
  \bibinfo{series}{EPTCS}} \bibinfo{volume}{161}, pp. \bibinfo{pages}{60--73},
  \doi{10.4204/EPTCS.161.8}.
\newblock \bibinfo{note}{Full version accepted for publication at
  \emph{Information and Computation}}.

\bibitemdeclare{inproceedings}{conf/lics/FinkbeinerS05}
\bibitem{conf/lics/FinkbeinerS05}
\bibinfo{author}{Bernd \surnamestart Finkbeiner\surnameend} \&
  \bibinfo{author}{Sven \surnamestart Schewe\surnameend}
  (\bibinfo{year}{2005}): \emph{\bibinfo{title}{Uniform Distributed
  Synthesis}}.
\newblock In: {\sl \bibinfo{booktitle}{{LICS} 2005}},
  \bibinfo{publisher}{{IEEE} Computer Society}, pp. \bibinfo{pages}{321--330},
  \doi{10.1109/LICS.2005.53}.

\bibitemdeclare{article}{journals/sttt/FinkbeinerS13}
\bibitem{journals/sttt/FinkbeinerS13}
\bibinfo{author}{Bernd \surnamestart Finkbeiner\surnameend} \&
  \bibinfo{author}{Sven \surnamestart Schewe\surnameend}
  (\bibinfo{year}{2013}): \emph{\bibinfo{title}{Bounded synthesis}}.
\newblock {\sl \bibinfo{journal}{{STTT}}}
  \bibinfo{volume}{15}(\bibinfo{number}{5-6}), pp. \bibinfo{pages}{519--539},
  \doi{10.1007/s10009-012-0228-z}.

\bibitemdeclare{article}{journals/acta/FridmanP14}
\bibitem{journals/acta/FridmanP14}
\bibinfo{author}{Wladimir \surnamestart Fridman\surnameend} \&
  \bibinfo{author}{Bernd \surnamestart Puchala\surnameend}
  (\bibinfo{year}{2014}): \emph{\bibinfo{title}{Distributed synthesis for
  regular and contextfree specifications}}.
\newblock {\sl \bibinfo{journal}{Acta Inf.}}
  \bibinfo{volume}{51}(\bibinfo{number}{3-4}), pp. \bibinfo{pages}{221--260},
  \doi{10.1007/s00236-014-0194-x}.

\bibitemdeclare{article}{journals/tocl/GastinS13}
\bibitem{journals/tocl/GastinS13}
\bibinfo{author}{Paul \surnamestart Gastin\surnameend} \&
  \bibinfo{author}{Nathalie \surnamestart Sznajder\surnameend}
  (\bibinfo{year}{2013}): \emph{\bibinfo{title}{Fair Synthesis for Asynchronous
  Distributed Systems}}.
\newblock {\sl \bibinfo{journal}{{ACM} Trans. Comput. Log.}}
  \bibinfo{volume}{14}(\bibinfo{number}{2}), p.~\bibinfo{pages}{9},
  \doi{10.1145/2480759.2480761}.

\bibitemdeclare{article}{journals/fmsd/GastinSZ09}
\bibitem{journals/fmsd/GastinSZ09}
\bibinfo{author}{Paul \surnamestart Gastin\surnameend},
  \bibinfo{author}{Nathalie \surnamestart Sznajder\surnameend} \&
  \bibinfo{author}{Marc \surnamestart Zeitoun\surnameend}
  (\bibinfo{year}{2009}): \emph{\bibinfo{title}{Distributed synthesis for
  well-connected architectures}}.
\newblock {\sl \bibinfo{journal}{Formal Methods in System Design}}
  \bibinfo{volume}{34}(\bibinfo{number}{3}), pp. \bibinfo{pages}{215--237},
  \doi{10.1007/s10703-008-0064-7}.

\bibitemdeclare{article}{journals/fmsd/KupfermanPV09}
\bibitem{journals/fmsd/KupfermanPV09}
\bibinfo{author}{Orna \surnamestart Kupferman\surnameend}, \bibinfo{author}{Nir
  \surnamestart Piterman\surnameend} \& \bibinfo{author}{Moshe~Y. \surnamestart
  Vardi\surnameend} (\bibinfo{year}{2009}): \emph{\bibinfo{title}{From liveness
  to promptness}}.
\newblock {\sl \bibinfo{journal}{Formal Methods in System Design}}
  \bibinfo{volume}{34}(\bibinfo{number}{2}), pp. \bibinfo{pages}{83--103},
  \doi{10.1007/s10703-009-0067-z}.

\bibitemdeclare{inproceedings}{conf/lics/KupfermanV01}
\bibitem{conf/lics/KupfermanV01}
\bibinfo{author}{Orna \surnamestart Kupferman\surnameend} \&
  \bibinfo{author}{Moshe~Y. \surnamestart Vardi\surnameend}
  (\bibinfo{year}{2001}): \emph{\bibinfo{title}{Synthesizing Distributed
  Systems}}.
\newblock In: {\sl \bibinfo{booktitle}{{LICS} 2001}},
  \bibinfo{publisher}{{IEEE} Computer Society}, pp. \bibinfo{pages}{389--398},
  \doi{10.1109/LICS.2001.932514}.

\bibitemdeclare{inproceedings}{conf/icalp/MadhusudanT01}
\bibitem{conf/icalp/MadhusudanT01}
\bibinfo{author}{Parthasarathy \surnamestart Madhusudan\surnameend} \&
  \bibinfo{author}{Pazhamaneri~Subramaniam \surnamestart
  Thiagarajan\surnameend} (\bibinfo{year}{2001}):
  \emph{\bibinfo{title}{Distributed Controller Synthesis for Local
  Specifications}}.
\newblock In: {\sl \bibinfo{booktitle}{{ICALP} 2011}}, {\sl
  \bibinfo{series}{LNCS}} \bibinfo{volume}{2076},
  \bibinfo{publisher}{Springer}, pp. \bibinfo{pages}{396--407},
  \doi{10.1007/3-540-48224-5\_33}.

\bibitemdeclare{inproceedings}{DBLP:conf/fsttcs/MohalikW03}
\bibitem{DBLP:conf/fsttcs/MohalikW03}
\bibinfo{author}{Swarup \surnamestart Mohalik\surnameend} \&
  \bibinfo{author}{Igor \surnamestart Walukiewicz\surnameend}
  (\bibinfo{year}{2003}): \emph{\bibinfo{title}{Distributed Games}}.
\newblock In \bibinfo{editor}{Paritosh~K. \surnamestart Pandya\surnameend} \&
  \bibinfo{editor}{Jaikumar \surnamestart Radhakrishnan\surnameend}, editors:
  {\sl \bibinfo{booktitle}{FSTTCS 2003}}, {\sl \bibinfo{series}{LNCS}}
  \bibinfo{volume}{2914}, \bibinfo{publisher}{Springer}, pp.
  \bibinfo{pages}{338--351}, \doi{10.1007/978-3-540-24597-1\_29}.

\bibitemdeclare{inproceedings}{Pnueli77}
\bibitem{Pnueli77}
\bibinfo{author}{Amir \surnamestart Pnueli\surnameend} (\bibinfo{year}{1977}):
  \emph{\bibinfo{title}{The temporal logic of programs}}.
\newblock In: {\sl \bibinfo{booktitle}{FOCS 1977}}, \bibinfo{publisher}{IEEE},
  pp. \bibinfo{pages}{46--57}, \doi{10.1109/SFCS.1977.32}.

\bibitemdeclare{inproceedings}{conf/focs/PnueliR90}
\bibitem{conf/focs/PnueliR90}
\bibinfo{author}{Amir \surnamestart Pnueli\surnameend} \& \bibinfo{author}{Roni
  \surnamestart Rosner\surnameend} (\bibinfo{year}{1990}):
  \emph{\bibinfo{title}{Distributed Reactive Systems Are Hard to Synthesize}}.
\newblock In: {\sl \bibinfo{booktitle}{{FOCS} 1990}},
  \bibinfo{publisher}{{IEEE} Computer Society}, pp. \bibinfo{pages}{746--757},
  \doi{10.1109/FSCS.1990.89597}.

\bibitemdeclare{article}{journals/ipl/Schewe14}
\bibitem{journals/ipl/Schewe14}
\bibinfo{author}{Sven \surnamestart Schewe\surnameend} (\bibinfo{year}{2014}):
  \emph{\bibinfo{title}{Distributed synthesis is simply undecidable}}.
\newblock {\sl \bibinfo{journal}{Inf. Process. Lett.}}
  \bibinfo{volume}{114}(\bibinfo{number}{4}), pp. \bibinfo{pages}{203--207},
  \doi{10.1016/j.ipl.2013.11.012}.

\bibitemdeclare{inproceedings}{conf/lopstr/ScheweF06}
\bibitem{conf/lopstr/ScheweF06}
\bibinfo{author}{Sven \surnamestart Schewe\surnameend} \&
  \bibinfo{author}{Bernd \surnamestart Finkbeiner\surnameend}
  (\bibinfo{year}{2006}): \emph{\bibinfo{title}{Synthesis of Asynchronous
  Systems}}.
\newblock In: {\sl \bibinfo{booktitle}{{LOPSTR} 2006}}, {\sl
  \bibinfo{series}{LNCS}} \bibinfo{volume}{4407},
  \bibinfo{publisher}{Springer}, pp. \bibinfo{pages}{127--142},
  \doi{10.1007/978-3-540-71410-1\_10}.

\bibitemdeclare{article}{journals/tcs/Zimmermann13}
\bibitem{journals/tcs/Zimmermann13}
\bibinfo{author}{Martin \surnamestart Zimmermann\surnameend}
  (\bibinfo{year}{2013}): \emph{\bibinfo{title}{Optimal bounds in parametric
  {LTL} games}}.
\newblock {\sl \bibinfo{journal}{Theor. Comput. Sci.}} \bibinfo{volume}{493},
  pp. \bibinfo{pages}{30--45}, \doi{10.1016/j.tcs.2012.07.039}.

\bibitemdeclare{inproceedings}{DBLP:journals/corr/Zimmermann15a}
\bibitem{DBLP:journals/corr/Zimmermann15a}
\bibinfo{author}{Martin \surnamestart Zimmermann\surnameend}
  (\bibinfo{year}{2015}): \emph{\bibinfo{title}{Parameterized Linear Temporal
  Logics Meet Costs: Still not Costlier than {LTL}}}.
\newblock In \bibinfo{editor}{Javier \surnamestart Esparza\surnameend} \&
  \bibinfo{editor}{Enrico \surnamestart Tronci\surnameend}, editors: {\sl
  \bibinfo{booktitle}{GandALF 2015}}, {\sl \bibinfo{series}{{EPTCS}}}
  \bibinfo{volume}{193}, pp. \bibinfo{pages}{144--157},
  \doi{10.4204/EPTCS.193.11}.

\end{thebibliography}

\end{document}